\documentclass[11pt]{article}
\usepackage{amsthm,amsmath,amsfonts,amssymb,mathrsfs,url} 
\usepackage[pdftex]{graphicx}
\usepackage{algorithm}
\usepackage{algorithmic}
\newcommand{\remove}[1]{}
\newcommand{\hide}[1]{}
\def\func#1{\textrm{\bf{\sc{#1}}}}
\newtheorem{thm}{Theorem}[section]

\newtheorem{lemma}[thm]{Lemma}
\newtheorem{definition}[thm]{Definition}

\def\codewidth{6in} 

\renewcommand{\remove}[1]{}

\newcommand{\comments}[1]{}

\def\draft{1}   

\ifnum\draft=1 
    \def\ShowAuthNotes{1}
\else
    \def\ShowAuthNotes{0}
\fi

\ifnum\ShowAuthNotes=1
\newcommand{\authnote}[2]{{ \footnotesize \bf{\color{red}[#1's Note: {\color{blue}#2}]}}}
\else
\newcommand{\authnote}[2]{}
\fi

\begin{document}

\title{\textbf{Characterization and Construction of a Family of Highly Symmetric Spherical Polyhedra with Application in Modeling Self-Assembling Structures}}
\author{Muhibur Rasheed and Chandrajit Bajaj\\~\\ Center for Computational Visualization \\ Institute of Computational Engineering and Sciences \\ The University of Texas  \\ Austin, Texas 78712}

\date{} 
\maketitle
\thispagestyle{empty}

\begin{abstract}
The regular polyhedra have the highest order of 3D symmetries and are exceptionally attractive templates for (self)-assembly using minimal types of building blocks, from nano-cages and virus capsids to large scale constructions like glass domes. However, they only represent a small number of possible spherical layouts which can serve as templates for symmetric assembly. In this paper, we formalize the necessary and sufficient conditions for symmetric assembly using exactly one type of building block. All such assemblies correspond to spherical polyhedra which are edge-transitive and face-transitive, but not necessarily vertex-transitive. This describes a new class of polyhedra outside of the well-studied Platonic, Archimedean, Catalan and and Johnson solids. We show that this new family, dubbed \emph{almost-regular} polyhedra, can be parameterized using only two variables and provide an efficient algorithm to generate an infinite series of such polyhedra. Additionally, considering the almost-regular polyhedra as templates for the assembly of 3D spherical shell structures, we developed an efficient polynomial time shell assembly approximation algorithm for an otherwise NP-hard geometric optimization problem.
\end{abstract}

\section{Introduction}
\label{sec:intro}
Regular polyhedra are combinatorial marvels. They are simultaneously isogonal (vertex-transitive), isotoxal (edge-transitive), and isohedral (face-transitive). In this context, transitivity means that the vertices (or edges or faces) are congruent to each other; and for any pair of vertices (or edges or faces), there exists a symmetry preserving transformation of the entire polyhedron which isometrically maps one to the other. Transitivity plays a vital role is assembly, especially self-assembly. For instance, if we consider each face as a building block, then face transitivity indicates that a single type of block is sufficient to form a shell-like structure. Similarly, edge-transitivity indicates that there is exactly one way to put any two blocks together. 

Viruses, natures smallest organisms, utilize this simplicity by encoding the blueprint (RNA/DNA) for a single type of protein. These proteins can attach within copies of the same proteins in specific ways. The attaching or binding rules is such that when sufficient concentrations of building blocks are available, they combine to assemble into a spherical shell, called a capsid. These spherical shells have icosahedral symmetry. However, unlike a icosahedron, which only has 20 faces, more than half of the viruses have capsids which are formed by a much larger number of proteins. Caspar and Klug \cite{caspar62} first addressed the layout of virus capsids using triangular tiles. A similar class of assembly is seen in fullerene like particles, with 12 pentagonal and many hexagonal faces, which was first characterized by Goldberg \cite{Goldberg_1937}. 

Recently, several researchers have developed efficient constructions and parameterizations of Goldberg-like particles \cite{Hu_Qiu_2008,Schwerdtfeger_Wirz_Avery_2013,Fowler_Rogers_2001}. For instance, Deng et al.\ \cite{Deng_Yu_2012} studied extensions of Goldberg's construction to other platonic solids, but their study is not exhaustive in characterization and enumeration of all the possible cases. In another recent work, Schein and Gayed \cite{Schein_Gayed_2014} developed a numerical optimization scheme that takes a Goldberg-like polyhedra, which by construction does not have planar faces and is not always convex, and produces strictly convex polyhedra while preserving the edge-lengths. However, the resulting polyhedra no longer have any face-transitivity. Also, even though the edges have the same length, they are not strictly congruent as their neighboring faces are different- hence making such polyhedra unsuitable for using as layouts for assembly. 

In this paper, we introduce the \emph{almost-regular} polyhedra, which preserve the global polyhedral symmetry while offer a denser packing with local symmetries, thus ensuring that only one type of building block still suffices. Topologically, these polyhedra are face-transitive and edge-transitive; and has at most two distinct types of vertices (each vertex is transitive to the other vertices of the same type). In some cases, which we have characterized, the polyhedra may become non-convex and numerical optimization is required to keep edge-lengths equal and make the faces as congruent as possible. Note that even the non-convex polyhedra of this family remain spherical (i.e., any ray emanating from the centroid of the polyhedra will intersect it exactly once).

We believe that our results will greatly impact research in several areas including the field of nano-materials which can self-assemble to create nano-structures with desirable properties. For example, gold nanorods for cancer imaging and therapy~\cite{Chen_2005,XiaGold2014}, virus capsids and protein-cages for targeted drug delivery \cite{Shi_2010,Smith_2013,Steinmetz_2009,Shang_2012}. Advanced scientific computation techniques, to explore and automatically predict possible nano-structures that can be formed symmetrically by one type of building block (eg. an engineered protein) would accelerate the development of new nano-shell structures. Our theoretical groundwork would greatly support such extensive computational techniques. For instance, we show that using our symmetry characterization, a symmetric assembly of $n$ particles can be predicted using a algorithm whose running time is only polynomial in $n$, even though assembly prediction is an NP-hard optimization problem. 

\section{Background}

The boundary of a convex polyhedron is homeomorphic to a spherical tiling. While the space of all convex polyhedra is not enumerable (uncountably infinite), the sub-classes which exhibit one or more types of symmetry and congruency conditions are enumerable. Typical congruency conditions considered in this context are face transitivity, edge transitivity and vertex transitivity. A polyhedron is face-transitive (isohedral) if all faces of the polyhedron are congruent and are transitive. In other words, there exists a symmetry transformation of the entire polyhedron which would map any specific face, A, onto another specific face, B. Similarly, it is edge-transitive (isotoxal) if all edges of the polyhedron are congruent and transitive in the same sense as above; and vertex-transitive (isogonal) if all vertices of the polyhedron are congruent and transitive in the same sense as above.  

Depending on which of the above properties are satisfied, we have the following classes of polyhedra-
\begin{itemize}
 \item Johnson solid: Each face is a regular polygon. But the polyhedron does not satisfy any of the transitivity properties. There are exactly 92 such solids- all of them convex. The subclass of Johnson solids which have only equilateral triangular faces are called Deltahedra.
 \item Catalan solid: Duals of the Archimedean solids. Catalan solids are convex and face-transitive; but not edge-transitive of vertex transitive.
 \item Semiregular (also called uniform) polyhedra: These are vertex-transitive and each face is a regular polygon (of 2 or more different types). Examples of such polyhedra are the 13 Archimedean solids, and infinite series of convex prisms and anti-prisms with regular polygons as the two parallel faces.
 \item Quasiregular polyhedra: This class have vertex and edge transitivity. It is easy to show that they can have at most two types of regular faces. There are only two such polyhedra- cuboctahedron and icosidodecahedron.
 \item Regular polyhedra: A regular polyhedron is face, edge and vertex transitive; and has only one type of regular face. There are only 5 regular polyhedra in 3 dimensions. These are the Platonic solids: Tetrahedron, Cube, Octahedron, Icosahedron and Dodecahedron.
\end{itemize}

These special classes are all enumerated and there are only a finite number of them, except for the infinite family of prisms and anti-prisms (which are also enumerable as they map directly to the set of natural numbers). 

\paragraph{A note on symmetry}
A Symmetry group consist of a set of symmetry operations (i.e. transformations which map an object to itself) such that the set is closed under the composition (one transformation followed by another) operation. The actions of the polyhedral symmetry groups can be expressed as pure rotations around different axes through the center of the polyhedron, which we assume to be at the origin without loss of generality. We generally refer to the points of intersection of these axes with the polyhedron as the locations of symmetry, and refer to the axes as the symmetry axes. For instance, the octahedron have 6 axes of 4-fold rotational symmetry\footnote{$n$-fold rotational symmetry, also referred to as the symmetry order $n$, means that a rotation by $2\pi/n$ maps the polyhedron to itself} going through the four vertices, 8 axes of 3-fold rotational symmetry going through the centers of the faces, and 12 axes of 2-fold rotational symmetry going through the centers of the edges.

\subsection{The \emph{Almost-Regular} Polyhedra and Their Duals}
We define the family of \emph{almost-regular} polyhedra as all polyhedra which have global polyhedral symmetry, have congruent regular faces, and is face and edge transitive. The first condition means that any such polyhedron must have exactly the same number of rotational symmetry axes with the same symmetry orders as any specific regular polyhedron. We refer to these as the \emph{global} symmetry axes and locations. We shall refer to these global symmetry axes as gv-symmetry axes, ge-symmetry axes and gf-symmetry axes respectively for axes of symmetry going through, respectively, the vertices, edge-centers and face-centers of the regular polyhedron. Additionally, due to the congruency conditions, all vertices, faces and edges of an \emph{almost-regular} polyhedron must have locations of \emph{local} symmetry. Local symmetry operations map the vertices, edges, faces immediately neighboring the location of local symmetry to themselves, but may or may not map the remaining parts of the polyhedron to itself. These axes will be referred to as lv-, le-, and lf-symmetry axes.

\begin{figure}[h!]
\centering
\includegraphics[width=0.85\linewidth]{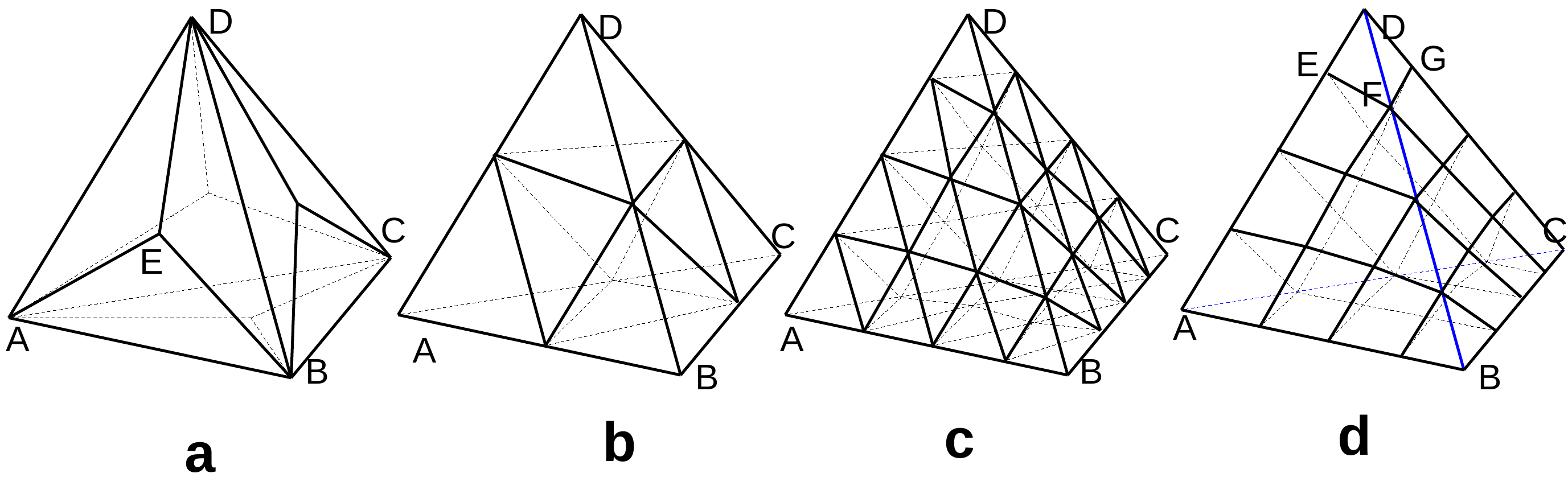}
\caption[Illustration of the criterion for almost-regular polyhedra]{ \textbf{Illustration of the criterion for almost-regular polyhedra}. The polyhedra shown in \textbf{(a)} and \textbf{(d)} are not almost-regular. In \textbf{(a)}, the global symmetries are preserved, but the local symmetries are not (for example, it is not locally 2-fold symmetric around the center of DE). In \textbf{(d)}, local symmetries are intact (note that DEFG and other creased faces are considered a single face), but the global 3-fold symmetries are not (for example, around the center of ABD).}
\label{fig:assemblytheory:quasi}
\end{figure}

While this class bear some similarity with the Catalan solids, an important distinction is that the Catalan solids allow non-symmetric faces (hence, are not isotoxal), as long as all the faces are congruent (for example, the solid in Figure \ref{fig:assemblytheory:quasi}(a) is a Catalan solid, but is not almost-regular.

Note that, the transitivity properties that make regular polyhedra suitable for assembly, is preserved in almost-regular polyhedra, but now the class is richer and more importantly can model structures with more than 20 building blocks.

The dual of an almost-regular polyhedron has regular faces, is isotoxal and isogonal; but may or may not be isohedral (they can hae at most two types of faces). The closest known family is the semi-regular polyhedra (duals of Catalan solids) which are also isotoxal and isogonal. But the semi-regular polyhedra, which includes the 13 Archimedean solids and the family of prisms with regular faces, have at least 2 types of faces.

\subsection{Related Prior Work}
The construction scheme, that we propose in the next section, closely follows the one proposed by Goldberg \cite{Goldberg_1937}. The family of polyhedra generated by the Goldberg construction rule \cite{Goldberg_1937} are fullerene like structures. Fullerene like structures have icosahedral symmetry (symmetry group of the icosahedron), and consists of many hexagonal faces and exactly 12 pentagonal faces. The soccer ball is the smallest example of such structure. See Figure \ref{fig:assemblytheory:construction} for an illustrative description of the construction.

\begin{figure}[h!]
\centering
\includegraphics[width=0.85\linewidth]{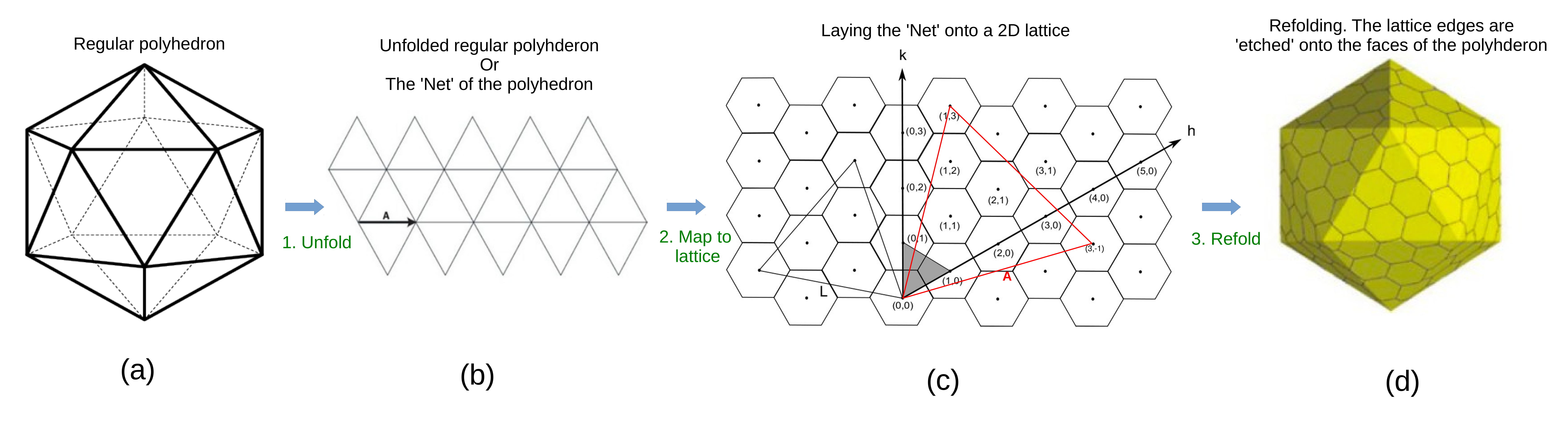}
\caption[Illustration of the Goldberg construction]{ \textbf{Illustration of the Goldberg construction}. The Goldberg construction involves unfolding an icosahedron (see \textbf{(a)} and \textbf{(b)}), and then mapping the unfolded icosahedron onto a 2D hexagonal lattice scaled and oriented such that all corners of the unfolded icosahedron (its original vertices) falls on the centers of some hexagon of the grid (some example scale and orientations (for one triangle only) are shown \textbf{(c)}). Finally, the icosahedron is folded back, along with the hexagonal grid etched onto its faces. For example, for the scaling and orientation of the red triangle in \textbf{(c)}, would result in the tiled icosahedron shown in \textbf{(d)}. Notice that the new polyhedron has exactly 12 regular pentagonal faces where the icosahedral vertices originally were, and many regular hexagonal faces.}   
\label{fig:assemblytheory:construction}
\end{figure}

Caspar and Klug \cite{caspar62} proposed a similar approach, but using a triangular lattice, instead of a hexagonal one, and required that the corners of an edge of the unfolded icosahedron falls on the vertices of the lattice (Figure \ref{fig:assemblytheory:CKdemo}). Since, the triangular lattice is simply the dual of the hexagonal lattice, the mapping is essentially the same. But the refolded polyhedron now has only regular triangular faces. It has 12 vertices where 5 such faces are incident, and many vertices where 6 faces are incident- the first set are exactly the original vertices of the icosahedron. Notice that this polyhedron is exactly the dual of the one constructed using Goldberg's method (shown in Figure \ref{fig:assemblytheory:construction}(d), and also overlayed in Figure \ref{fig:assemblytheory:CKdemo}(b)).

\begin{figure}[h!]
\centering
\includegraphics[width=0.85\linewidth]{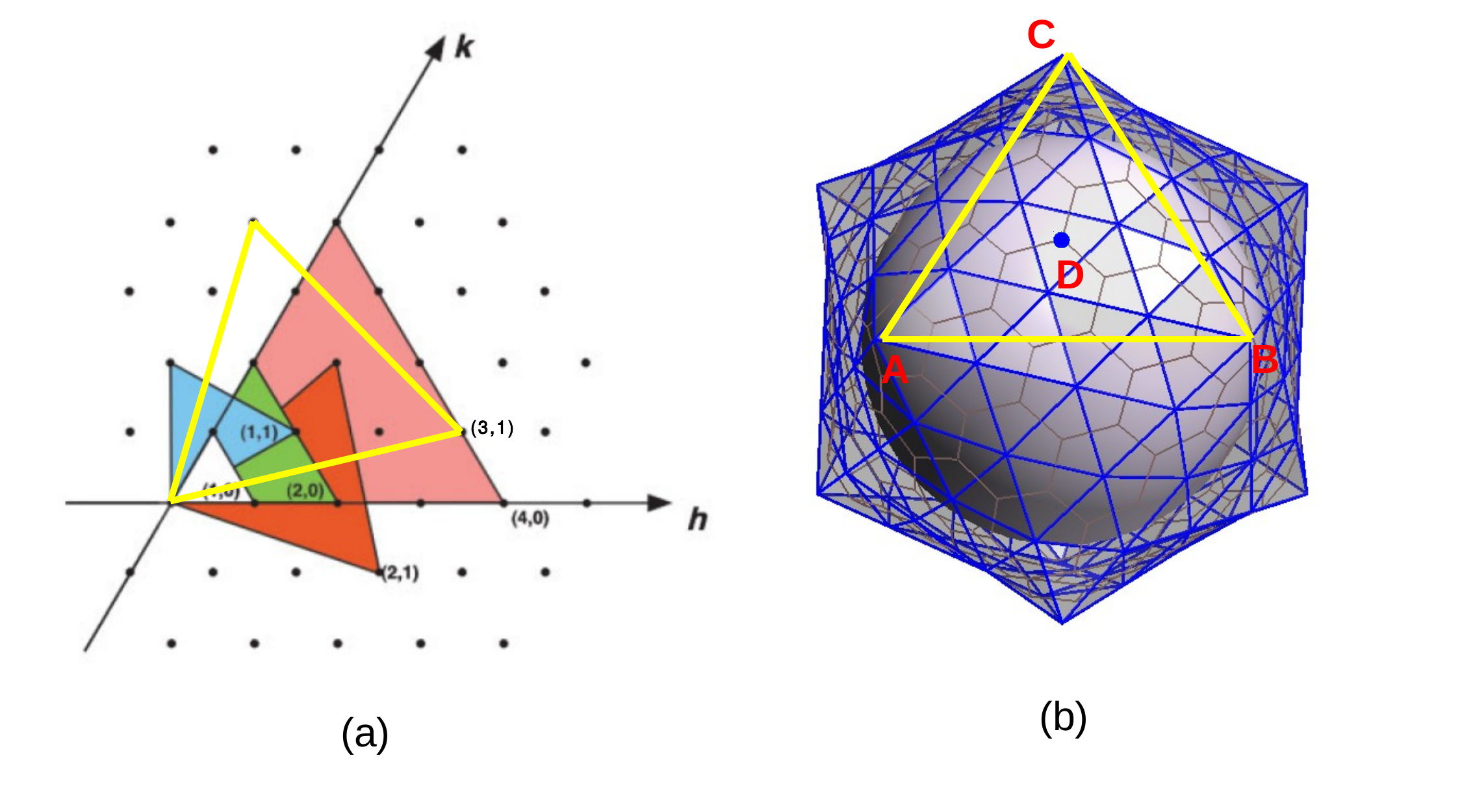}
\caption[Illustration of the Caspar-Klug construction]{ \textbf{Illustration of the Caspar-Klug construction}. The Caspar-Klug construction involves unfolding an icosahedron onto a triangular lattice scaled and oriented such that all corners of the unfolded icosahedron (its original vertices) falls on the vertices of the grid (some example scale and orientations (for one triangle only) are shown \textbf{(a)}). Then, the icosahedron is folded back, along with the grid etched onto its faces. For example, for the scaling and orientation of the yellow triangle in \textbf{(a)}, would result in the tiled icosahedron shown in \textbf{(b)}.}   
\label{fig:assemblytheory:CKdemo}
\end{figure}

Polyhedra produced by Caspar and Klug's construction method are \emph{almost-regular}, and the ones produced by Goldberg's are duals of \emph{almost-regular}. But notice that both Goldberg, and Caspar and Klug focused only on the icosahedral case and considered a specific unfolding onto specific 2D lattices, and hence only covers a fraction of the possible almost-regular polyhedra and their duals. Separately, Pawley \cite{Pawley_1961} studied other ways of wrapping different polyhedra using different lattices from the wallpaper group. However, Pawley did not provide any theoretical characterization of the factors related to the possibility or impossibility of such wrappings. We address this issue in the following section.

\section{Characterizing All Possible \emph{Almost-Regular} Polyhedra}
\label{sec:assemblytheory:enumeration}
Both Goldberg and Caspar-Klug constructions can be expressed as unfolding a regular polyhedron onto a 2D lattice and then refolding it with the lattice etched onto its faces. Pawley's wrapping idea is equivalent. We call these procedures the \emph{unfold-etch-refold} method. Here, we prove the conditions that must be satisfied to produce almost-regular polyhedra using the \emph{unfold-etch-refold} idea for any regular solid, unfolded in any way, onto any 2D lattice.

Shepherd's conjecture \cite{Shephard_1975} states that all convex polyhedra have a non self-overlapping planar unfolding with only edge-cuts. This conjecture is not proved or disproved yet for all possible convex polyhedra. However, it is true for the set of special classes we are interested in. Hence, in principle it is possible to unfold one such polyhedron and lay it down on a 2D grid, use the grid to draw tiles of the unfolded polyhedron, and then fold it back up to get a tiled polyhedron. However, every polyhedron actually have many unfolding. For example, icosahedron have 43380 unique unfoldings. Caspar and Klug's construction produced almost-regular polyhedra using 1 such unfolding, but it is not clear whether other unfoldings would also produce similar \emph{almost-regular} polyhedra, or different types of \emph{almost-regular} polyhedra, or not be \emph{almost-regular}. To address this question, we characterize the relationship of the local and global symmetries of the almost-regular polyhedra, and the etched polyhedra (henceforth called \emph{tiling}) produced using unfold-etch-refold construction.

First of all, we prove that the lattice onto which the polyhedron is unfolded must be regular.

\begin{lemma}
\label{lemma:assemblytheory:regulargrid}
 The polyhedra generated by an unfold-etch-refold using any regular polyhedra and unfolded in any way, cannot be \emph{almost-regular} if a non-regular grid/lattice is used.
\end{lemma}

\begin{proof}
 The unfold-etch-refold essentially wraps the lattice/grid over a regular polyhedron. Hence, vertices, edges and faces of the regular polyhedron are suppressed, and a new set of vertices, edges and faces appear (all of which belong to the lattice). Now, the \emph{local} symmetry condition of the almost-regular polyhedron requires that every face be symmetric around its center and be congruent to each other. This cannot be satisfied if the lattice itself was not regular, or not symmetric around some points. 
\end{proof}

There are exactly 3 regular lattices in 2D- the square lattice, the triangular lattice and the hexagonal lattice. The square lattice has 4-fold rotational symmetries at each vertex and face-center, the triangular lattice has 6-fold and 3-fold symmetries at each vertex and face-center; and the hexagonal lattice has 3-fold and 6-fold symmetries at each vertex and face-center. All of them have 2-fold symmetry on the center of each edge.

Among regular polyhedra, the tetrahedron, the octahedron and the icosahedron have 3-fold symmetries at face-centers and respectively 3, 4 and 5-fold symmetries at vertices. The cube has 4-fold symmetries at vertices and face-centers. The dodecahedron has 3-fold and 5-fold symmetries at vertices and face-centers. 

Now we prove another lemma addressing global symmetry conditions at gf-, ge-, gv-symmetry axes.

\begin{lemma}
\label{lemma:assemblytheory:localglobal}
 To satisfy gf-symmetry conditions, all gf-axes must go through a point of the lattice that have $cn$-fold rotational symmetry where $n$ is the order of rotational symmetry around the gf-symmetry axes of the regular polyhedron and $c$ is a positive integer.
\end{lemma}

\begin{proof}
 Recall that an almost-regular must have exactly the same number of rotational symmetry axes with the same symmetry orders as any specific regular polyhedron, such that there exists a rigid body transformation that perfectly aligns these axes to those of the regular polyhedron. Since, the unfold-etch-refold is a wrapping, the expected locations and axes of global symmetry of the new tiled/etched polyhedra, and the underlying regular polyhedra are already aligned. For example, in Figure \ref{fig:assemblytheory:CKdemo}(b), A, B and C are locations of 5-fold global symmetry and D is a location of 3-fold local symmetry. Let, $D$ be the location of one axis of symmetry going through the face of the regular polyhedron and $n$ be its symmetry order. Depending on the chosen unfolding and mapping, a particular gf-symmetry axis can have the following three cases-

 \begin{itemize}
  \item If a gf-symmetry axis goes through a vertex of the tiled polyhedra, then the tiled polyhedra can be $n$-fold symmetric around that axis only if the lattice have $cn$-fold rotational symmetry around its vertices.
  \item If a gf-symmetry axis goes through the center of a face of the tiled polyhedra, then the tiled polyhedra can be $n$-fold symmetric around that axis only if the lattice face is $cn$-regular.
  \item If a gf-symmetry does not go through a vertex or a face-center, then the tiled polyhedra can not be $n$-fold rotational symmetric around the axis irrespective of the symmetry of the regular grid. 
 \end{itemize}

 Hence, the lemma is proved for any unfolding/mapping.
\end{proof}

A regular 2D grid/lattice can be parameterized using two vectors $e_1, e_2$, and a fixed origin $O$ as follows. Let, $O$ be a vertex of the lattice and that $O$ has at least 2 neighbors $U$ and $V$ such that the vectors $O\rightarrow U$ and $O\rightarrow V$ are not collinear. Then, defining $e_1 = O\rightarrow U$ and $e_2 = O\rightarrow V$, every point of the lattice can be expressed as linear combinations $he_1 + ke_2$ where $h$ and $k$ are integers. This defines a coordinate system $\mathcal{L}$ with $O$ as the origin and $e_1, e_2$ as the primary axes and a co-ordinate $(h,k)$ representing points on the 2D plane, such that if both $h$ and $k$ are integers then, the point lies on the lattice (is a lattice vertex).

\begin{figure}[h!]
\centering
\includegraphics[width=0.85\linewidth]{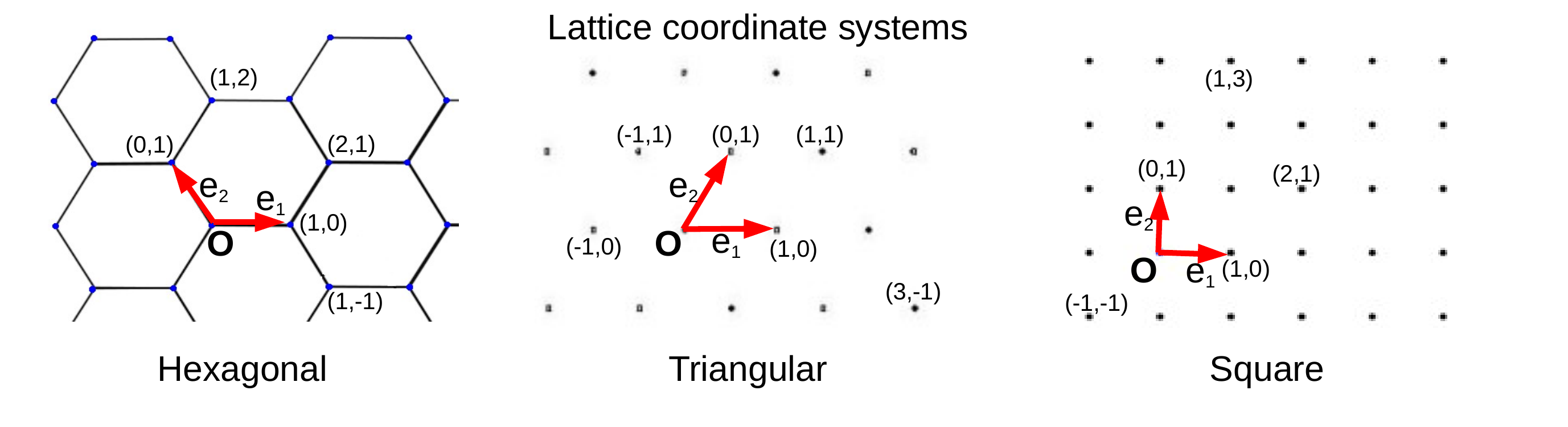}
\caption[Coordinate systems for the 2D regular lattices]{ \textbf{Coordinate systems for the 2D regular lattices}. In the figures, $O$ is the origin and $e_1$ and $e_2$ are two coordinate axes. Any point $(x,y)$ can be reached by vector $xe_1+y_e2$ from the origin. Coordinates of some of the lattice points are shown.}   
\label{fig:assemblytheory:coords}
\end{figure}

\begin{lemma}
 \label{lemma:assemblytheory:sufficient0}
 A lattice described in the $\mathcal{L}$ coordinate system is 2-fold rotationally symmetric around a point $(x,y)$ for the following cases-
 \begin{itemize}
  \item Triangular lattice: if and only if both $2x$ and $2y$ are integers
  \item Square lattice: if and only if both $2x$ and $2y$ are integers
  \item Hexagonal lattice: if and only if both $2x$ and $2y$ are odd integers, or $x=2(y+1)-3k$ where $x$, $y$ and $d$ are integers.
 \end{itemize}
\end{lemma}

\begin{proof}
 Triangular lattices have 2-fold (or multiples of 2-fold) symmetries around their vertices and edge-centers, both of which satisfy the condition that $2x$ and $2y$ are integers. Also, no other point satisfy these conditions.
 Square lattices have 2-fold (or multiples of 2-fold) symmetries around their vertices, face-centers and edge-centers, all of which satisfy the condition that $2x$ and $2y$ are integers. Also, no other point satisfy these conditions.
 Hexagonal lattices have 2-fold (or multiples of 2-fold) symmetries around their face-centers and edge-centers. The edge-centers satisfy the condition that both $2x$ and $2y$ are odd integers. The face centers are integer solutions of $x$, $y$ for the family of lines $x=2(y+1)-3d$ where $d$ is an integer. No other point satisfy either of these conditions.
\end{proof}

\begin{lemma}
 \label{lemma:assemblytheory:sufficient1}
 If a face $\mathcal{T}$ of a regular polyhedron is mapped to a 2D regular lattice in the following ways, then the part of lattice inside $\mathcal{T}$ is $n$-fold rotational symmetric around the center of $\mathcal{T}$, where $n$ is the rotational symmetry around the gf-axes of the regular polyhedron, and the set of faces intersected by each edge of $\mathcal{T}$ is $2$-fold rotationally symmetric around the center of the edge (location of the ge-axes of the regular polyhedron).
 \begin{enumerate}
  \item If $\mathcal{T}$ belongs to either a tetrahedron, an octahedron or a icosahedron and is mapped to triangular lattice, such that all corners have integer coordinates.
  \item If $\mathcal{T}$ belongs to either a tetrahedron, a octahedron or a icosahedron and is mapped to a hexagonal lattice such that all corners fall on face centers of the lattice (have integer coordinates $(x,y)$ such that $x=2(y+1)-3d$ where $x$, $y$ and $d$ are integers).
  \item If $\mathcal{T}$ belongs to a cube and is mapped to a square lattice, such that either all the corners fall on vertices of the lattice, or all the corners fall on face-centers of the lattice.
 \end{enumerate}
\end{lemma}

\begin{proof}
 First, we show that such mapping is possible.
 \begin{enumerate}
  \item Consider the case when $\mathcal{T}$ is an equilateral triangular face being placed on a triangular lattice. Without loss of generality, we assume that one corner is placed at the origin and another at $(h,k)$ where $h$ and $k$ are integers. Then, it is trivial to show that the third point can be at $(h+k,-h)$.
  \item Consider the case when $\mathcal{T}$ is an equilateral triangular face being placed on a hexagonal lattice. Note that the hexagonal lattice is simply the dual of the triangular lattice, hence a mapping that puts corners of $\mathcal{T}$ on vertices of the triangular lattice would put corners of $\mathcal{T}$ on face-centers of the hexagonal lattice.
 \item Consider the case when $\mathcal{T}$ is a square being placed on a square lattice. Without loss of generality, we assume that one corner is placed at the origin and another at $(h,k)$ where $h$ and $k$ are integers. Then, it is trivial to show that the other points can be placed at $(h-k, h+k)$ and $(-k, h)$.
 \end{enumerate}

 When an equilateral triangular face $\mathcal{T}=ABC$ is placed on a triangular lattice such that the corners are at $(0,0), (h,k)$ and $(h+k,-h)$, the the center $O$ of $\mathcal{T}$ is at $(\frac{2h+k}{3},\frac{-h+k}{3})$ which is at a face center, or at a vertex (if $h=k$, or $k=-2h$). Hence $O$ falls on a location of 3-fold or 6-fold symmetry. Hence, gf-symmetry is satisfied. The center $(x,y)$ of any edge would satisfy the condition that $2x$ and $2y$ are integers, and hence falls on a location that have 2-fold (or 6-fold) symmetry and satisfies the ge-condition.

 Similar arithmetic can be applied to prove the theorem for the remaining two cases.
\end{proof}

We shall refer to the mapping described in Lemma \ref{lemma:assemblytheory:sufficient1} a \emph{compatible mapping}.

\begin{lemma}
\label{lemma:assemblytheory:sufficient2}
 For any unfolding of a regular polyhedron onto a regular lattice, if any face $\mathcal{T}$ is compatibly mapped, then all faces are also compatibly mapped. And, the etching inside each face are congruent.
\end{lemma}

\begin{proof}
 For any unfolding, there must be at least another face adjacent to $\mathcal{T}$ and shares an edge with it. Let that edge be $AB$. According to Lemma \ref{lemma:assemblytheory:sufficient0}, for each other point $P$ belonging to $\mathcal{T}$, there exists a point $P'$ produced by rotating $P$ around the center of $AB$ by 180 degrees. Also, if $P$ was on a vertex (or a face-center), then $P'$ will also be the same. Since only a rigid body motion is applied, the new face $\mathcal{T}' = ABP'\ldots$ will also be regular and be congruent to $\mathcal{T}$ (i.e. it is exactly the unfolded face which was neighboring $\mathcal{T}$ along $AB$). Also, for any point $X$ in $\mathcal{T}$, the same transformation would map it to a point $X'$ inside $\mathcal{T}'$ such that the location of $X$ with respect to the corners of $\mathcal{T}$ is the same as the location of $X'$ with respect to the corners of $\mathcal{T}'$. Hence the etching inside $\mathcal{T}$ and $\mathcal{T}'$ are congruent. By propagation of the same argument, all unfolded faces have corners at integer coordinates and are congruent, irrespective of the unfolding.
\end{proof}

\begin{lemma}
\label{lemma:assemblytheory:sufficient3}
 For any unfolding of a regular polyhedron onto a regular lattice, if any face $\mathcal{T}$ is compatibly mapped, then the resulting polyhedron will be \emph{almost-regular} or a dual.
\end{lemma}

\begin{proof}
 The polyhedron generated by unfold-etch-refold method have local symmetry, due to the lattice being regular.
 
 Lemma \ref{lemma:assemblytheory:sufficient1} showed that within the face $\mathcal{T}$, it satisfies global symmetry around the gf-axes. Even after the faces are folded back into the complete polyhedron, the congruency of the all the face (Lemma \ref{lemma:assemblytheory:sufficient2}) guarantees that the remaining faces would also map to each other (including the etching inside them). 

 Lemma \ref{lemma:assemblytheory:sufficient2} showed that all faces (and edges) are congruent, and also that the etching is 2-fold symmetric around the edges of the face $\mathcal{T}$. Hence when folded back, the etchings from a neighboring face $\mathcal{T}'$ will match up perfectly (intersect the shared edge of $\mathcal{T}$ and $\mathcal{T}'$ at exactly the same points), and in case there are fractions of a lattice-face inside $\mathcal{T}$, exactly its complement will show up on the other side (inside $\mathcal{T}'$) of the shared edge, thereby all etched-faces are complete. Topologically, the etched-faces that cross the face boundaries and the ones that do not, are identical. 

 After folding back, the corners of the faces $\mathcal{T}, \mathcal{T}', \ldots$ meet at a point. Note that since all faces are congruent (Lemma \ref{lemma:assemblytheory:sufficient2}) and all corners of each face are also symmetric to each other (Lemma \ref{lemma:assemblytheory:sufficient1}), gv-symmetry is satisfied. Moreover, when $\mathcal{T}$ was mapped such that the corners fell on vertices of the lattice, then the gv-axis is surrounded by exactly $n$ congruent and regular etched-faces, where $n$ is the rotational symmetry of the gv-axis, and not the rotational symmetry around the lattice-vertices. For instance, mapping a tetrahedron onto a triangular lattice will produce 4 vertices where 3 edges are incident, and many more (depending on the scale of the mapping) where 6 edges are incident. Hence, two types of vertices appear on the polyhedron, the ones coinciding with the gv-axis and the ones that do not. On the other hand, if $\mathcal{T}$ was mapped such that the corners fell on face-centers of the lattice, then the gv-axis is surrounded by exactly a regular polyhedron with $nc$-fold symmetry, where $n$ is the rotational symmetry of the gv-axis and $c$ is an integer. Other etched faces will simply depend on the regularity of the lattice. For example, if an icosahedron is mapped to a hexagonal grid, then there would be 12 pentagonal faces and many hexagonal faces. 

 Hence, the polyhedron is \emph{almost-regular} if the corners of the faces fell on lattice vertices, or dual if the corners fell on face-centers. Tetrahedron mapped to triangular grid's face-centers is an exception, where the dual construction also produce \emph{almost-regular} polyhedron. 
\end{proof}

Figure \ref{fig:assemblytheory:primaldualconst} shows a few examples of constructing \emph{almost-regular} polyhedron and their duals by mapping a face of regular polyhedron on regular lattices in a compatible way. Note that the etched-faces that cross an edge of $\mathcal{T}$ are geometrically not identical to the ones that do not. The ones crossing the boundary have a crease inside them, or if they are flattened, they are no longer regular. This is addressed in the Section \ref{sec:assemblytheory:optimization}.

\begin{figure}[h!]
\centering
\includegraphics[width=0.85\linewidth]{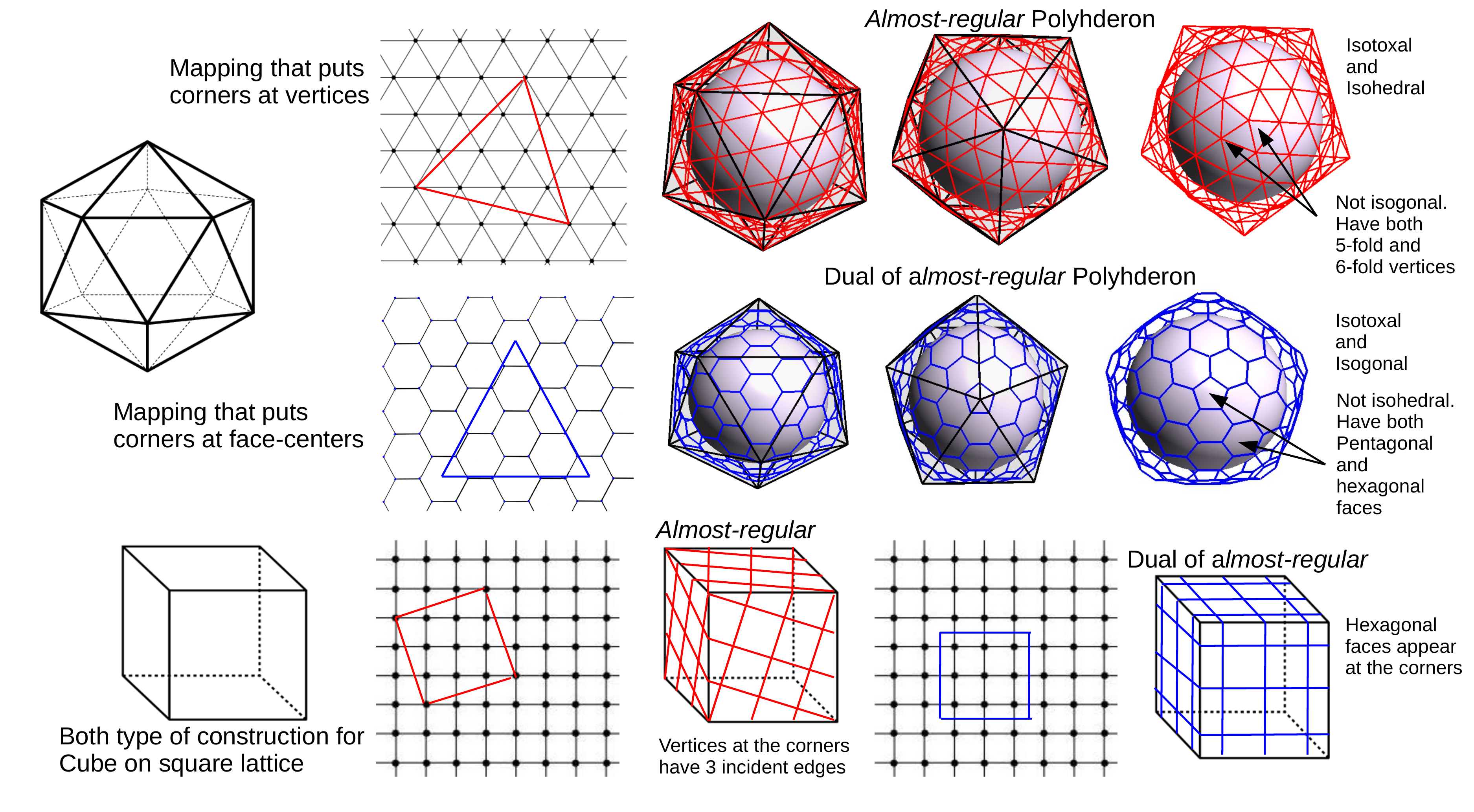}
\caption[Illustration of the constructing \emph{almost-regular} polyhedron and their duals]{ \textbf{Illustration of the constructing \emph{almost-regular} polyhedron and their duals}. Top row shows how placing corners of a polyhedral face on vertices of a compatible lattice produces an \emph{almost-regular} polyhedron. The black lines show the original polyhedron, and the red lines show the etching/tiling induced by the lattice. The second row shows and example of placing the corners at face centers and producing duals of \emph{almost-regular} polyhedron. Finally, the bottom row shows examples of both primal and dual construction using square lattice.}
\label{fig:assemblytheory:primaldualconst}
\end{figure}

Now we show that any deviation from a compatible mapping results in a violation of some global symmetry condition.

\begin{lemma}
\label{lemma:assemblytheory:necessary1}
 If a regular polyhedron is mapped to a regular lattice without fulfilling the compatible mapping conditions, then the resulting polyhedron cannot be \emph{almost-regular}.
\end{lemma}

\begin{proof}
 If some corners of a face $\mathcal{T}$ fall on lattice vertices (or face-centers) and other don't, then the etching inside the face cannot be symmetric around the center of the face. Hence the polyhedron is not \emph{almost-regular}.

 If none of the corners of $\mathcal{T}$ fall on lattice vertices (or face-centers), and the center of the face is not on a lattice vertex or center of a lattice face, then again the polyhedron is not \emph{almost-regular} (by Lemma \ref{lemma:assemblytheory:localglobal}). 

 Finally, we consider the case where none of the corners of $\mathcal{T}$ fall on lattice vertices (or face-centers), but the center $O$ of $\mathcal{T}$ is on a lattice vertex or a face-center. This would not violate the symmetry around the gf-symmetry axes. But we shall show that it would violate the symmetry around the ge-symmetry. Since, we are mapping a regular polyhedron onto a compatible lattice, there exists a mapping that would place a face $\mathcal{T}'$ such that the center falls at $O$, and the corners lie at integer coordinates. We can generate $\mathcal{T}$ by scaling and/or rotating $\mathcal{T}'$ around $O$. Now, we shall prove that any scaling or rotation of $\mathcal{T}$ that moves the vertices off lattice-vertices violate ge-symmetry. The proof depends on the geometry of coordinate system and here we shall only show if for the cube mapped to the square lattice case. Other cases follow the same pattern of proof.

 Without loss of generality, we assume that the center $O$ is the origin, and we focus on one edge $AB$ of the face $\mathcal{T}$ with integer coordinates $(x_1,y_1)$ and $(x_2,y_2)$. Hence, the midpoint $C(\frac{x_1+x_2}{2},\frac{y_1+y_2}{2})$ falls on a face-center, vertex, or edge-center. After rotating around $O$ by $\theta$ degrees, the new location of the center will be $C'(\frac{(x_1+x_2)cos\theta - (y_1+y_2)sin\theta}{2},\frac{(x_1+x_2)sin\theta + (y_1+y_2)cos\theta}{2})$. Hence, $C'$ can be on a 2-fold location if and only if the numerators are integers, which can only happen if $\theta$ is a multiple of $\pi/2$ and in that case, $\mathcal{T}'$ coincides with $\mathcal{T}$. 

 Now, if $\mathcal{T}$ was scaled by $s$ then the new position of the edge-center $C$ would be $C(\frac{(x_1+x_2)s}{2},\frac{(y_1+y_2)s}{2})$, and again the numerators are integers iff $s$ is an integer. And if $s$ is an integer, then the corners of $\mathcal{T}'$ would also be on lattice vertices or face-centers.

 Hence, scaling or rotating $\mathcal{T}$ around its center $O$ will not keep the $C$ on a face-center, vertex, or edge-center unless the corners of the transformed face $\mathcal{T}'$ fall on lattice-vertices. Finally, the fact that a regular lattice is not 2-fold symmetric around any point other than the face-centers, vertices, or edge-centers completes the proof.
\end{proof}

Finally, we conclude our characterization with the following theorem.

\begin{thm}
\label{thm:assemblytheory:characterization}
 The polyhedron generated by an unfold-etch-refold is \emph{almost-regular} if and only if a compatible mapping of a regular polyhedron onto an unfold-etch-refold compatible lattice is performed.
\end{thm}

\begin{proof}
The proof follows from Lemmas \ref{lemma:assemblytheory:regulargrid}, \ref{lemma:assemblytheory:localglobal}, \ref{lemma:assemblytheory:sufficient3} and \ref{lemma:assemblytheory:necessary1}.
\end{proof}

\subsection{Parametrization}
\label{sec:assemblytheory:combinatorics}
For the sake of simplicity of presentation, the following discussion is focused solely on mapping the icosahedron onto triangular lattices. Other compatible mappings can be discussed in the same manner with almost no difference in the theorems/lemmas presented here except for minor changes in counting. The choice to focus on the icosahedral case is primarily due to two reasons: first, it has the highest level of symmetry among the regular polyhedra which have a compatible mapping, and second, it has applications in modeling viruses, fullerenes etc. 

Let $\mathcal{L}$ be a lattice with origin $O$ and axes $H$ and $K$. Any point in the lattice is expressed using coordinates $(h,k)$ where both $h$ and $k$ are integers.

\begin{lemma}
 Assuming that one corner $A$ of the face $\mathcal{T}$ of the polyhedron is mapped to the origin $O$ of the lattice (or the nearest face-center for dual constructions). Then specifying the position of another compatibly placed point $B(h,k)$ is sufficient to parametrize the entire mapping. 
\end{lemma}

\begin{proof}
 Since $\mathcal{T}$ is regular, there are exactly 2 possible ways to define the other points of $\mathcal{T}$. Lemma \ref{lemma:assemblytheory:sufficient1} showed that both of these choices will result in a compatible mapping as long as $A$ and $B$ are also compatibly mapped. Also by Lemma \ref{lemma:assemblytheory:sufficient2}, these two are congruent. So, any one of them can be chosen arbitrarily.
\end{proof}

We mention the following lemma whose proof is immediate.

\begin{lemma}
 Topology of any \emph{almost-regular} polyhedron or its dual can be parametrized using a tuple $<\mathcal{P},\mathcal{L},h,k>$, where $\mathcal{P}$ is a regular polyhedron, $\mathcal{L}$ is a lattice represented using two axes, and $h$ and $k$ are integers.
\end{lemma}

\subsubsection{Combinatorial Results}
We consider the case where $\mathcal{P}$ is an icosahedron whose symmetry group will is $I$, and $\mathcal{L}$ is the triangular lattice which will be denoted as $\mathcal{L}^3$. Now, we discuss some properties of the lattice.

\begin{definition}
We define each triangle of the lattice $\mathcal{L}^3$ as a small triangle\index{Small Triangle} and use $t$ to denote such a triangle. Let us define a triple $<i,j,k>$ where $i$ and $j$ are integers and $k \in \{+,-\}$. Let the triangle produced by the intersections of the lines $h=i$, $k=j$ and $h+k=i+j+1$ (having the vertices $(i,j), (i+1,j)$ and $(i, j+1)$) be denoted $t_{ij+}$. Similarly, the triangle denoted $t_{ij-}$ has vertices $(i,j), (i+1,j-1)$ and $(i+1, j)$, and is produced by the intersections of the lines $h=i+1$, $k=j$ and $h+k=i+j$. 
\label{def:smallt}
\end{definition}

The proof of the following lemma is immediate from this definition.

\begin{lemma}
$t_{i_1j_1k_1}$ coincide with $t_{i_2j_2k_2}$ if and only if $i_1 = i_2$, $j_1 = j_2$ and $k_1 = k_2$. For any small triangle in $\mathcal{L}^3$, there exists a triple $<i,j,k>$ such that $t_{ijk}$ represents that small triangle.
\end{lemma}

Through etching, the triangular lattice $\mathcal{L}^3$ produces a tiling of a face $\mathcal{T}$ (which will be called a large triangle\index{Large Triangle} in this section) of $\mathcal{P}$ where each tile is a small triangle. Now we consider some properties of this tiling.

Assuming $A$ is at $(0,0)$, $B$ is $(h,k)$ such that $h$ and $k$ are integers, the tiling produced by $\mathcal{L}^3$ on $\mathcal{T}$ satisfies:
   \begin{itemize}
      \item The area of $\mathcal{T}$ is $\frac{\sqrt{3}}{4} (h^2+hk+k^2)$, which is equal to the area of $h^2+hk+k^2$ small triangles.
      \item In addition to $A, B$ and $C$, $\mathcal{T}$ includes exactly $\frac{h^2+hk+k^2-1}{2}$ more vertices of $\mathcal{L}^3$. Note that any vertex that lie on an edge of $\mathcal{T}$ is counted as half a vertex.
      \item Each edge of $\mathcal{T}$ is intersected by at most $2(h + k) - 3$ lines of the form $h=c$, $k=d$ and $h+k=e$, where $c,d$ and $e$ are integers.
      \item The number of small triangles intersected by any edge of $\mathcal{T}$ is at most $2(h+k-1)$.
   \end{itemize}

The following are some combinatorial properties of the overall tiled polyhedron:
   \begin{itemize}
      \item There are exactly $20(h^2+hk+k^2)$ small triangles, and the same number of local 3-fold axes.
      \item The 12 gf-symmetry axes are surrounded by 5 small triangles.
      \item There are exactly $10(h^2+hk+k^2-1)$ vertices (not lying on the gf-axes) with 6-fold local symmetry.
   \end{itemize}

Similar properties can easily be derived for other mappings as well. The important point to note is that not only the topology, but also the symmetry and combinatorics are also parameterized by only $h$ and $k$.

\section{Constructing all \emph{Almost-regular} Polyhedra}
In the previous sections we characterized the conditions that must be satisfied by a \emph{unfold-etch-refold} protocol to produce an \emph{almost-regular} polyhedron. The characterization immediately lends itself to efficient generation of families of such polyhedra whose topology can be parameterized using just two variables (discussed below). Furthermore, the symmetry at global and local levels lets us represent the geometry using a minimal set of points. Finally, we show how these properties lead to efficient optimization algorithms for constructing 3D shapes with spherical symmetries.

Note that given a point $P$ with coordinate $(i,j)$ inside $\mathcal{T}$, there exists two other points $Q$ and $R$ such that $P$, $Q$ and $R$ are 3-fold symmetric around the center $D$ of $\mathcal{T}$. The two points $Q$ and $R$ have coordinates $(h-i-j, k+i)$ and $(h+k+j,-h-i)$ respectively. This can be seen by noticing that stepping along the $H$ and $K$ axis by $i$ and $j$ units from $A(0,0)$ is $C^3$ symmetric (around $D$) to stepping in $-H+K$ and $-H$ directions by the same units from $B(h,k)$, and stepping in $-K$ and $H-K$ directions by the same units from $C$. We can further extend it to triangles and deduce the following.

\begin{lemma}
If $A(h_1, k_1)$, $B(h_2, k_2)$ and $C(h_3, k_3)$ are three points in the $HK$ coordinate system such that $h_1, h_2, h_3, k_1, k_2, k_3$ are integers and $ABC$ is an equilateral triangle whose centroid is $O$, then the small triangles $t_{h_1+i,k_1+j,\pm}$, $t_{h_2-i-j-1,k_2+i,\pm}$ and $t_{h_3+j,k_3-i-j-1,\pm}$ are $C^3$ symmetric around $O$.
\label{lem:trisym2}
\end{lemma}

Now, we define the minimal set of points or non-redundant set of points $\mathbb{S}$ such that no two points $s_i,s_j\in S$ are $C^3$ symmetric to each other around $D$, and all points in $\mathbb{S}$ lie inside or on $\mathcal{T}$. Clearly, $|\mathbb{S}| = \lceil\frac{h^2+hk+k^2}{3}\rceil$. Note that applying $C^3$ operations on $\mathbb{S}$ produces all points inside and on $\mathcal{T}$.

Now we are ready to specify a concrete algorithm for computing \emph{almost-regular} polyhedra (see Figure \ref{fig:ConstructAlmostRegularPolyhedron}).

\begin{figure}[h!] 
\centering
\begin{minipage}{\codewidth}
\begin{center}
 	\framebox{
		\begin{minipage}{\codewidth}
	 		{\scriptsize
	
	 		\noindent\func{TilingGen$(\mathcal{P},\mathcal{L},h,k)$}

			Constructs an \emph{almost-regular} polyhedron using compatible mapping of polyhedron $\mathcal{P}$ onto lattice $\mathcal{L}$, such that the scaling and combinatorics are specified by $h$,$k$

	 		\noindent
	 		\begin{enumerate}
				\item Assume that the lattice coordinate system is aligned with the Cartesian coordinate system such that the origins coincide and one of the axes is aligned to the X axis, and the other lies on the XY plane.
	 			\item Place one point $A$ at the origin $(0,0)$ of the lattice, a second point at $(h,k)$. Compute the other corners of the face $\mathcal{T}$. Note that we only need to know the number of vertices $n$ of $\mathcal{T}$.
				\item Compute the location of the center $D$ of the face $\mathcal{T}$.
				\item Let $\mathbb{T_C}$ be the set of cyclic symmetry operations around $D$, such that $|\mathbb{T_C}| = n$.
				\item Initialize empty set $\mathbb{S}$
				\item For each lattice point $p$ inside or on $\mathcal{T}$ do
				\item ~~~~ Add $p$ to $\mathbb{S}$ if and only none of the transformations in $\mathbb{T_C}$ applied to $p$ produces a point which is already in $\mathbb{S}$.
	 			\item Compute the transformation $T_{map}$ which maps the face $\mathcal{T}$ to a face of the polyhedron $\mathcal{P}$. $T_{map}$ is composed of $T_{map_T}T_{map_S}T_{map_A}$ such that $T_{map_A}$ translates $\mathcal{T}$ along the lattice to take $D$ to the origin $O$, $T_{map_S}$ is a scaling that resizes $\mathcal{T}$ to the size of the faces in $\mathcal{P}$, then $T_{map_T}$ is a translation along Z-axis by an amount equal to the distance from the center of $\mathcal{P}$ to a face-center.
	 			\item Let $\mathbb{T_P}$ be the set of global symmetry operations (from the symmetry group of $\mathcal{P}$.
			        \item Define a set of transformations $\mathbb{T}_{all} = \{T_2T_{map}T_1|T_2 \in \mathbb{T_P} \& T_1 \in \mathbb{T_S}\}$.
				\item All points on the almost-regular polyhedron is now generated by simply computing $\mathbb{T}_{all}(\mathbb{S})$.
			\end{enumerate}
	 		}
 
 		\end{minipage}
 	}
\end{center}
\end{minipage}
\vspace{-0.2cm}
\caption[\func{TilingGen}: Algorithm for constructing an \emph{almost-regular} polyhedron using compatible mapping]{\func{TilingGen}: Algorithm for constructing an \emph{almost-regular} polyhedron using compatible mapping}
\label{fig:ConstructAlmostRegularPolyhedron}
\vspace{-0.3cm}
\end{figure}

\begin{thm}
The algorithm \func{TilingGen} constructs a minimal geometric representation of the \emph{almost-regular} polyhedron in terms a set of points $\mathbb{S}$ embedded onto the XY plane and a set of 3D transformations $\mathbb{T}_{all}$.
\end{thm}
\begin{proof}
 Follows from the definition of $\mathbb{S}$ and Lemma \ref{lemma:assemblytheory:sufficient3}.
\end{proof}

Some polyhedron generated by applying \func{TilingGen} are shown in Figure \ref{fig:TilingGenResults}.

Note that if the tiles that cross the boundaries of a face $\mathcal{T}$ of $\mathcal{P}$ are not regular, they would look like they have a crease along the edge of the $\mathcal{T}$ by definition of the unfold-etch-refold technique. Tiles generated by this algorithm will also have the same problem and such tiles will be non-regular, and in some cases even non-planar. In the next section we address this issue.

\begin{figure}[h!]
\centering
\includegraphics[width=0.85\linewidth]{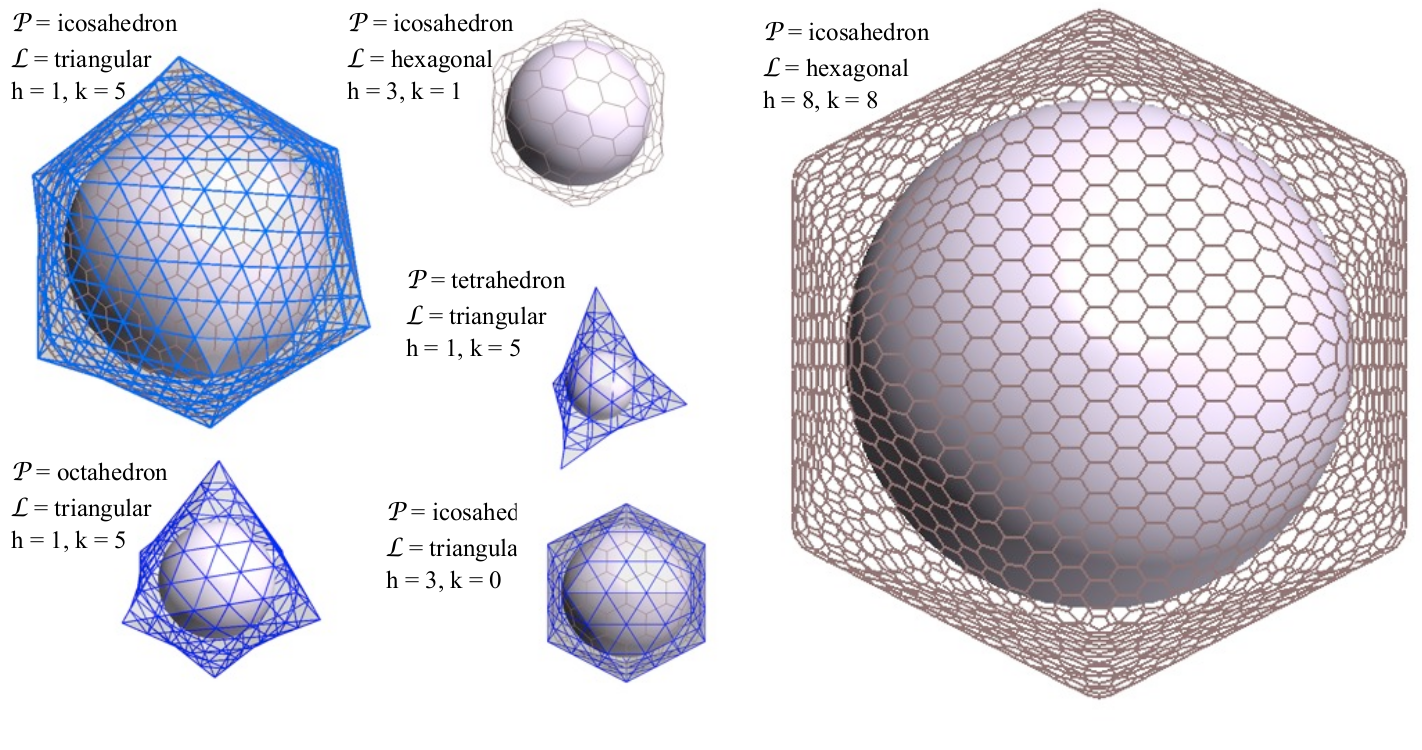}
\caption[Some polyhderon generated by applying \func{TilingGen}]{ \textbf{Some polyhderon generated by applying \func{TilingGen}}.}   
\label{fig:TilingGenResults}
\end{figure}

Some polyhedron generated by applying \func{TilingGen} are shown in Figure \ref{fig:TilingGenResults}.

Note that the tiles that cross the boundaries of a face $\mathcal{T}$ of $\mathcal{P}$ may have a crease along the edge of the $\mathcal{T}$. This is a result of the \emph{unfold-etch-refold} technique. Tiles generated by this algorithm will also have the same problem and such tiles will be non-regular, and in some cases even non-planar. 

\subsection{Curation of Tiles}
\label{sec:assemblytheory:optimization}
As mentioned before, sometimes the lattice faces, which corresponds to tiles/faces of the generated \emph{almost-regular} polyhedron, crosses the boundary of the polyhedral face $\mathcal{T}$ embedded on the lattice. During folding, these faces get warped. There can be exactly three types of scenarios (for mapping to a square or triangular lattice).

\begin{enumerate}
 \item If one corner of $\mathcal{T}$ is at $(0,0)$, and the other corner is at $(h,k)$ such that either $h=0$ or $k=0$, then no lattice face crosses the edges of $\mathcal{T}$, and no curation is required. (see Figure \ref{fig:assemblytheory:curation} top row).
 \item If one corner of $\mathcal{T}$ is at $(0,0)$, and the other corner is at $(h,k)$ such that $h=k$ then some lattice faces are exactly bisected by the edges of $\mathcal{T}$. The curation is quite trivial in this case. If $h=k$, the center of $\mathcal{T}$ would lie on a lattice vertex. Let the center be $D$, and the face $\mathcal{T}$ be $ABC$. Then, folding along $AD$, $BD$ and $CD$ will not warp any lattice face. Additionally, connecting $D$ to the centers of the neighboring polyhedral faces $\mathcal{T}$ would satisfy all global symmetry conditions as well. This folding will produce a base polytope which actually looks like the \emph{almost-regular} polyhedron with $h=k=1$. In fact, for any integer $i$, to polyhedron generated for $h=k=i$ is nothing but subdivisions of the faces of the $h=k=1$ polyhedron. Interestingly, in some cases, the new folding produces a polyhedron with base geometry like some Catalan solids, but will unlike Catalan solids, these will have regular faces and may be non-convex. For example, the $h=k=1$ polyhedron have the same topology as the pentakis dodecahedron (see Figure \ref{fig:assemblytheory:curation} middle row).
 \item If one corner of $\mathcal{T}$ is at $(0,0)$, and the other corner is at $(h,k)$ such that $h\neq k, \& h,k>0$; then, for all mappings on the hexagonal lattice, some lattice faces shall cross the edges of $\mathcal{T}$ in variable ways. There does not exist any folding which can avoid the crossing while maintaining global symmetry (as the lines will not meet at the center of the face $\mathcal{T}$). See Figure \ref{fig:assemblytheory:curation} bottom row. In this case, no exact solution exists, and we provide a numerical approximation which maximizes the regularity while ensuring that global symmetries are not violated (see below). 
\end{enumerate}

\begin{figure}[h!]
\centering
\vspace{-0.6cm}
\includegraphics[width=0.85\linewidth]{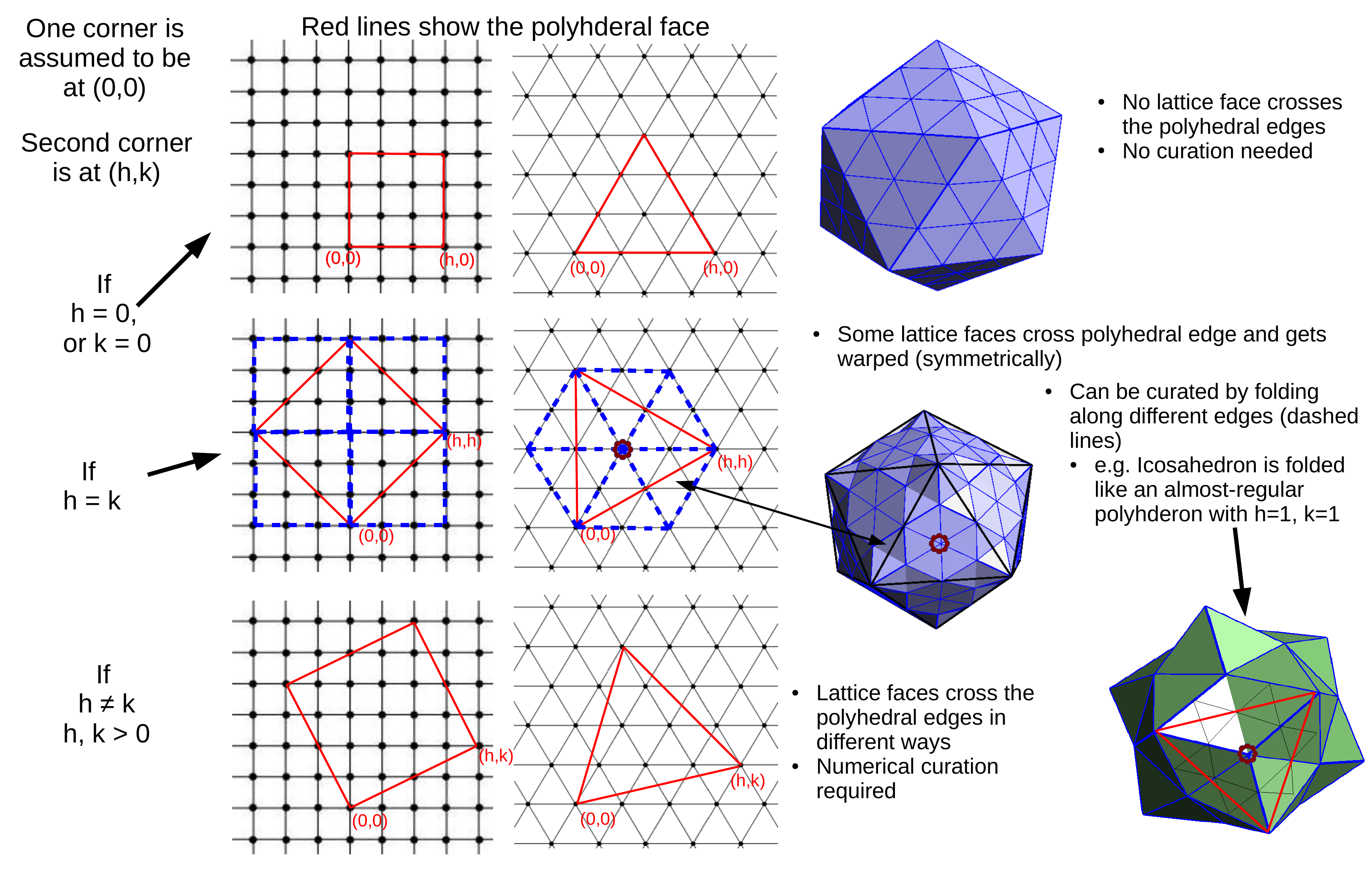}
\vspace{-0.3cm}
\caption[Different cases of warping of tiles, and their curations]{ \textbf{Different cases of warping of tiles, and their curations}.}   
\label{fig:assemblytheory:curation}
\vspace{-0.5cm}
\end{figure}

We should mention that recently Mannige and Brooks \cite{Mannige_Brooks_2010} explored such pseudo-irregularities in icosahedral tilings and suggested the existence of three classes depending on the values of $h$ and $k$. Here, we have generalized that to all \emph{almost-regular} polyhedra, and quantified the exact number of faces that gets warped.

\subsubsection{Curation as a Numerical Optimization Problem}
The goal of curation would be to make all the faces as regular as possible without violating global symmetry. A similar problem specifically for the family of polyhedra generated by mapping an icosahedron onto a hexagonal lattice was addressed in \cite{Schein_Gayed_2014}. In that work Schein and Gayed aimed to make all the hexagonal faces that cross the face boundary, and become creased/non-planar, into planar ones while keeping the edge lengths equal. They also showed that it is possible to provide an efficient numerical solution to the problem which ensures that no two hexagonal/pentagonal tiles lie on the same plane and the overall polyhedron is convex. However, the shapes of the hexagons are allowed to get distorted such that the angles are no longer equal, and may vary a lot within the same hexagon. Hence, the faces are no longer congruent (or even nearly congruent) to each other. This makes such a polyhedron non-amenable for modeling structures formed using a single type of building block, for instance viral capsids. In contrast, we want to maintain the congruence of the tiles as much as possible. 

When mapping a polyhedron onto a triangular lattice, the generated polyhedra falls under the class called deltahedra, polyhedra whose faces are all equilateral triangles. Even though there are an infinite family of deltahedra \cite{Trigg_1978} (our families are also infinite), it has been known since Freudenthal and van der Waerden's work \cite{Freudenthal_1947}, that there are exactly eight convex deltahedra, having 4, 6, 8, 10, 12, 14, 16 and 20 faces; among them only three are regular or have symmetries like the regular ones. So, our family of \emph{almost-regular} polyhedra cannot be convex and regular at the same time. We prioritized regularity. 

\begin{figure}
\centering
\vspace{-0.5cm}
\includegraphics[width=0.8\linewidth]{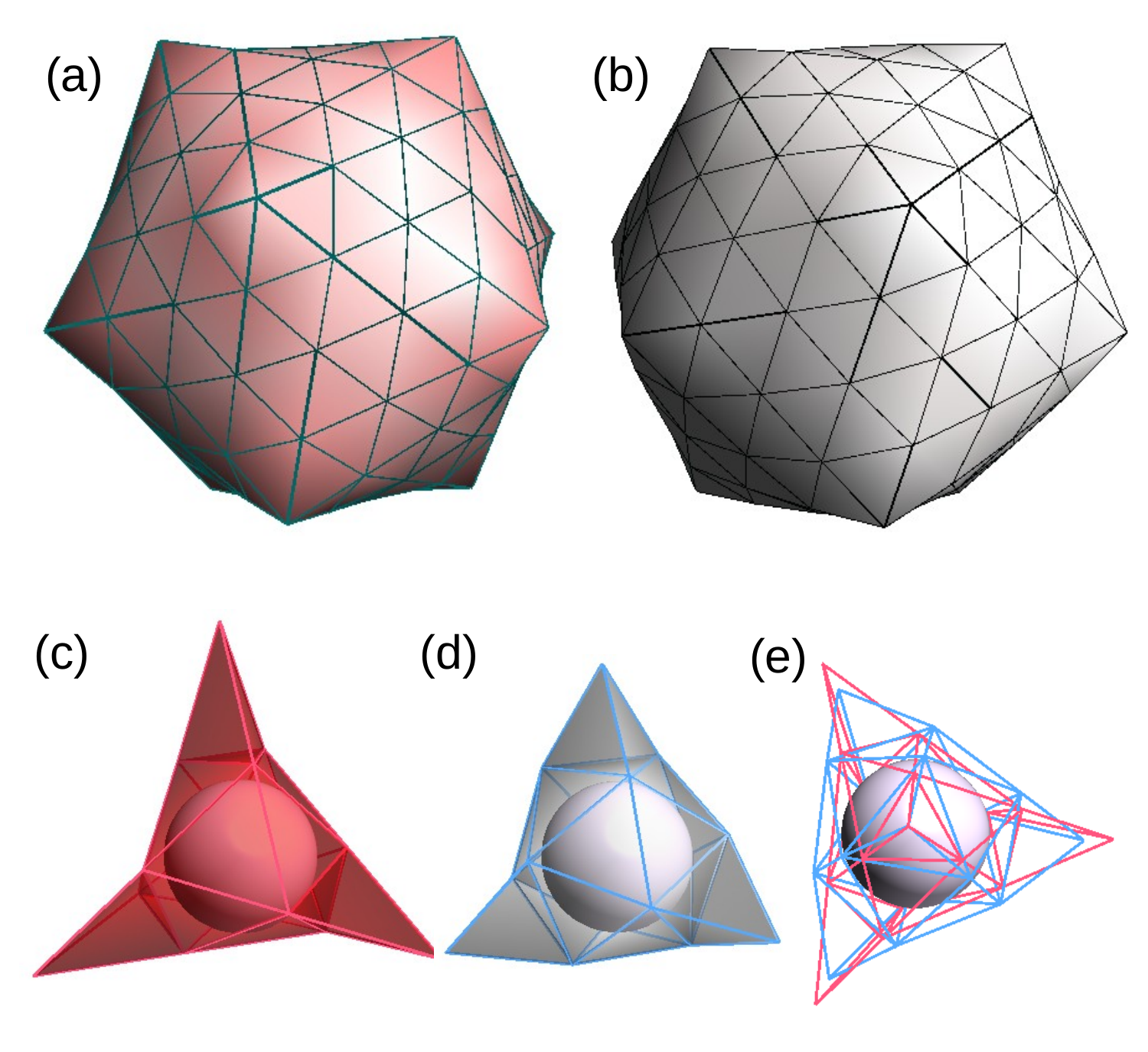}
\vspace{-0.5cm}
\caption[Numerical curation of warped faces]{ \textbf{Numerical curation of warped faces}. Top row: (a) shows a polyhedron with icosahedral symmetry and 260 tiles. The color of the triangles are determined by the ratio of the longest and the shortest edge of the tile. Ratios 1 to 1.3+ are colored using white to red gradient. The triangles at the corners and near the edges of the icosahedron have higher distortion. (b) shows the result of numerical curation. All the triangles are now regular, the worst ratio being 1.034, and the symmetry is preserved. Bottom row: (c-d) show a similar before-after figure for a smaller polyhedra, but one where the warping is more apparent. The numerical optimization brought down the ratio of the worst triangle from 2.13 to 1.004. (e) shows a superposition of the two states to highlight that points are updated symmetrically.}
\label{fig:assemblytheory:demo3}
\vspace{-0.5cm}
\end{figure}

\begin{itemize}
 \item Let the set of all points on the generated polyhedron be $\mathbb{S}_{all} = \mathbb{T}_{all}(\mathbb{S})$.
 \item Let $\mathbb{E}_1$ be the set of lattice/tile edges on the generated polyhedron
 \item Let $\mathbb{E}_2$ be the set of diagonals of all the tiles on the generated polyhedron
 \item Let, for any point $p \in \mathbb{S}_{all}$, the functions $s(p)$ and $t(p)$ returns respectively a point $q \in \mathbb{S}$ and a transformation $T \in \mathbb{T}_{all}$ such that $p = T(q)$.
 \item Let $dist(u,v)$ be the square of the Euclidean distance between two points.
\end{itemize}

In our calculations we shall only update the positions of the points in $\mathbb{S}$, and all other points $p \in \mathbb{S}_{all}$ on the polyhedron will be generated as $t(u)(s(u))$. This ensures that the points are always moved in a symmetric way with respect to the global symmetry axes. Hence global symmetry is never violated.

Let $\mathbb{S}^0$ be the initial positions of the points in $\mathbb{S}$. As we update the locations of the points in $\mathbb{S}$ in our algorithm, the squared displacement of each point $p \in \mathbb{S}$ will be defined as $\delta(p) = (dist(p,p^0))^2$, where $p^0 \in \mathbb{S}^0$ is the initial position of $p$. Also, the squared length of a line segment $e(u,v) \in \mathbb{E}_1 \cap \mathbb{E}_2$ will be computed as $dist(t(u)(s(u)),t(v)(s(v)))^2$ and be denoted $l(e)$. Let us also define $\mu_1 = \frac{1}{|\mathbb{E}_1|} \sum_{e \in \mathbb{E}_1} (l(e))$ and $\mu_2 = \frac{1}{|\mathbb{E}_2|} \sum_{e \in \mathbb{E}_2} (l(e))$. 

Finally, we define an energy function $\mathcal{F}(\mathbb{S})$ as follows:

$\mathcal{F}(\mathbb{S}) = \frac{1}{|\mathbb{E}_1|} (\sum_{e \in \mathbb{E}_1} (l(e) - \mu_1)) + \frac{1}{|\mathbb{E}_2|} (\sum_{e \in \mathbb{E}_2} (l(e) - \mu_2)) + \frac{1}{|\mathbb{S}|} (\sum_{p\in \mathbb{S}} \delta(p)).$

Now, we minimize the function $\mathcal{F}(\mathbb{S})$ over the positions of the points in $\mathbb{S}$.

This is clearly a quadratic optimization problem over $h^2+hk+k^2$ variables, and for most practical values of $h$ and $k$ it can be solved efficiently using any numerical solution techniques. We chose to use the limited memory variant of Broyden-Fletcher-Goldfarb-Shanno (BFGS) algorithm \cite{Byrd_Lu_Nocedal_Zhu_1995} due to its faster convergence rates. Also since first and second derivatives (hessian) of the energy function are straight-forward to compute, the numerical solution does not require finite-differences and heuristic based hessians, making the solution under BFGS more efficient and stable. Figure \ref{fig:assemblytheory:demo3} shows some examples of numerically curating warped faces.

\section{Constructing Shell Structures}
Viruses, as discussed before, have icosahedral symmetric shells formed by multiple copies of the same protein. While several existing works (e.g. \cite{berger94,zlotnick05,zandi05,Rapaport_2004,Bona_Sitharam_Vince_2011,Bahadur_Rodier_Janin_2007,Carrillo-Tripp_Brooks_Reddy_2008,Cheng_Brooks_2013}) leverage symmetry to analyze a given shell structure, we are the first to propose a generative algorithm to predict all such structures.

We propose that each virus shell is templated on a particular \emph{almost-regular} polyhedron. We consider the assmebly prediction problem where the number of proteins and the structure of a individual protein is known, and the structure of the whole shell is unknown and must be predicted. We propose the following generation algorithm (Figure \ref{fig:ShellGen}) to solve this problem. Note that, while the algorithm specifies icosahedral symmetry as that is the relevant one for viruses, it can easily be extended to handle other cases.

\begin{figure}[h!]
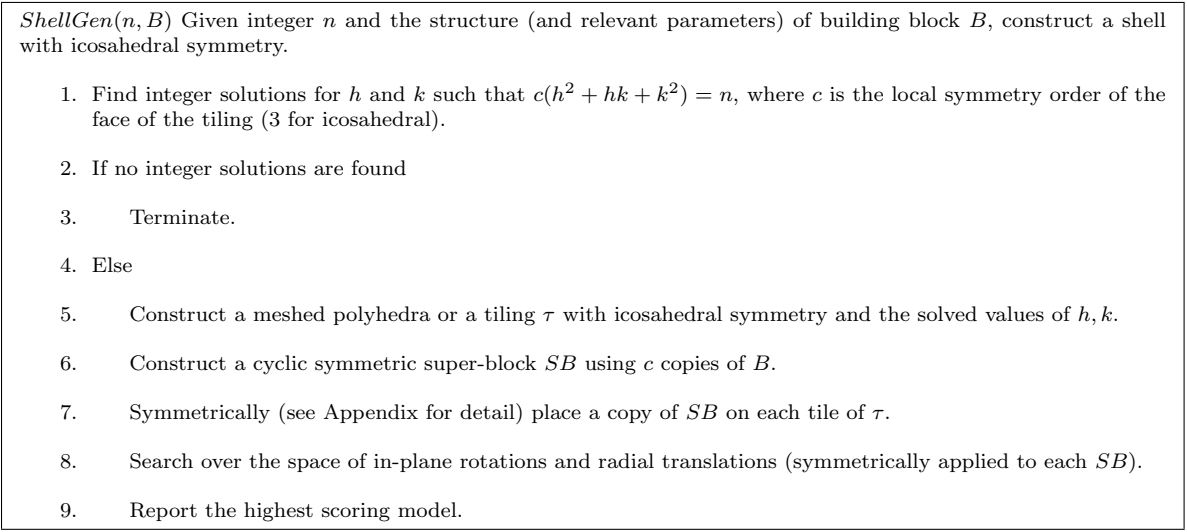
 
\centering
\begin{minipage}{\codewidth}
\begin{center}
 	\framebox{
		\begin{minipage}{\codewidth}
	 		{\scriptsize
	
	 		\noindent\func{$ShellGen(n,B)$}
			Given integer $n$ and the structure (and relevant parameters) of building block $B$, construct a shell with icosahedral symmetry.
			
	 		\noindent
	 		\begin{enumerate}
				\item Find integer solutions for $h$ and $k$ such that $c(h^2+hk+k^2)=n$, where $c$ is the local symmetry order of the face of the tiling (3 for icosahedral).
				\item If no integer solutions are found 
				\item ~~~~ Terminate.
	 			\item Else
				\item ~~~~ Construct a meshed polyhedra or a tiling $\tau$ with icosahedral symmetry and the solved values of $h,k$.
				\item ~~~~ Construct a cyclic symmetric super-block $SB$ using $c$ copies of $B$.
				\item ~~~~ Symmetrically (see Appendix for detail) place a copy of $SB$ on each tile of $\tau$. 
				\item ~~~~ Search over the space of in-plane rotations and radial translations (symmetrically applied to each $SB$).
				\item ~~~~ Report the highest scoring model.
			\end{enumerate}
	 		}
 
 		\end{minipage}
 	}
\end{center}
\end{minipage}
\vspace{-0.2cm}
\caption[\func{ShellGen}]{\func{ShellGen}: Algorithm for constructing a shell structure template upon an \emph{almost-regular} of requisite size.}
\label{fig:ShellGen}
\vspace{-0.3cm}
\end{figure}

Analysis of the algorithm follows.

\subsection{Datails of Decoration Rules for Constructing Thick Shell Structures}
We assume that the structures of tiles with the same internal symmetry (e.g. small triangle) remain the same across the shell. And all tiles are decorated by identical building blocks in a cyclic symmetric configuration. In other words, we assume a direct mapping between the internal symmetry of the tile and the set of building blocks used to decorate it. So for example, we decide to use a complex of three identical blocks in a $C^3$ configuration (henceforth called a c-tile) to decorate a small triangle of the layout, as opposed to using 3 independent blocks to decorate each corner of the c-tile.

\paragraph{Representation of a tile} Each tile is represented using five parameters $(\mathbf{u_1}, \mathbf{u_2}, f)$.
  \begin{itemize}  
    \item $\mathbf{u_1}$ is a unit vector pointing from the origin to the center of symmetry of the tile.
    \item $\mathbf{u_2}$ is a unit vector pointing from the center of symmetry of the tile to one representative corner of the tile.
    \item $f$ is the order of symmetry of the decoration to be placed inside the tile.
  \end{itemize}

\paragraph{Representation of a c-tile} Each $c$-tile is represented using five parameters $(\mathbf{v_1}, \mathbf{v_2}, \mathbf{c}, o)$.
    \begin{itemize}  
      \item $\mathbf{v_1}$ is a unit vector representing the symmetry axis.
      \item $\mathbf{v_2}$ is a unit vector orthogonal to the symmetry axis pointing to the centroid of one copy from the center of symmetry.
      \item $\mathbf{c}$ is the center of symmetry.
      \item $o$ is the order of symmetry.
    \end{itemize}

\paragraph{Decorating rules}
To decorate a tile $(\mathbf{u_1}, \mathbf{u_2}, f)$ using a $c$-tile $(\mathbf{v_1}, \mathbf{v_2}, \mathbf{c}, o)$, the following must be satisfied-
\begin{itemize}
 \item $o = f$.
 \item $c$ is translated such that it lies on the line $m{u_1}=0$ where $m$ is a scalar.
 \item Rotated such that $v_1$ aligns with $u_1$ and $v_2 \times u_2$ is parallel to $u_1$.
\end{itemize}

Further, we ensure the following-
\begin{itemize}
 \item If $C_i$ and $C_j$ are any two $c$-tiles of the same order,
   \begin{itemize}
      \item The $c$-tiles are translated in a symmetric way. i.e.\ $m_i = m_j$.
      \item The $c$-tiles are rotated in a symmetric way. i.e.\ $v_{2_i} \cdot u_{2_i} = v_{2_j} \cdot u_{2_j}$.
   \end{itemize}
\end{itemize}

The following result is immediate.

\begin{lemma}
 The decoration rules are necessary and sufficient for preserving local symmetry (internal to the tile) and global symmetry of the tiling.
\end{lemma}

Since the \func{ShellGen} algorithm satisfy the decoration rules, it correctly generates symmetric shell structures.

\subsection{Reults}

\begin{figure}[t!]
\centering
\vspace{-0.3cm}
 \includegraphics[width=0.9\linewidth]{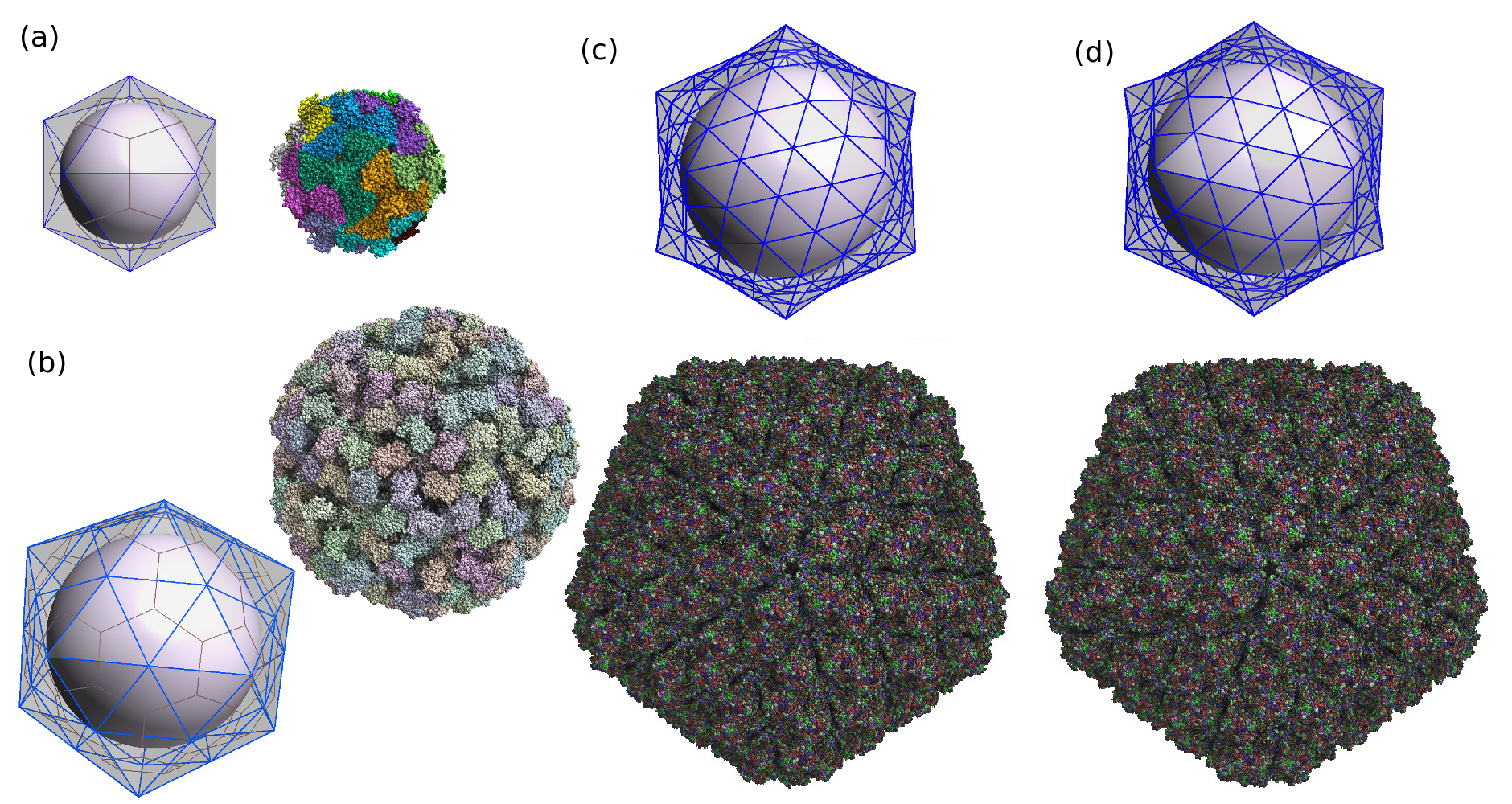}
\vspace{-0.2cm}
\caption[Predicted shell structures for known viruses] {\textbf{Predicted shell structures for known viruses.} (a) Predicted shell structure for Tobacco Necrosis Virus using a polyhedra with h=1, k=0. (b) Predicted shell structure for Nudaurelia Capensis Virus using a polyhedron with h=2, k=0. (c) Predicted shell structure for Rice Dwarf Virus outer shell using h=1, k=3. (d) Predicted structure for the same virus with h=3, k=1. All the predicted models have the correct inter-tile (or inter-protein) interfaces/contact in terms of geometry, and have the correct global and local topology; except, the one in (c) which has wrong chilarity).}
\vspace{-0.5cm}
\label{fig:assemblytheory:demo}
\end{figure}

\vspace{-0.5cm}
\subsubsection{Reproducing Known Shell Structures}
To verify the effectiveness of the tiling and decoration algorithms, we applied it to predict shell structures for some viruses for which the structure of the individual building blocks (proteins), as well as the entire shell is known. We show some examples here. 

\begin{figure}[t!]
\centering
 \includegraphics[width=0.85\linewidth]{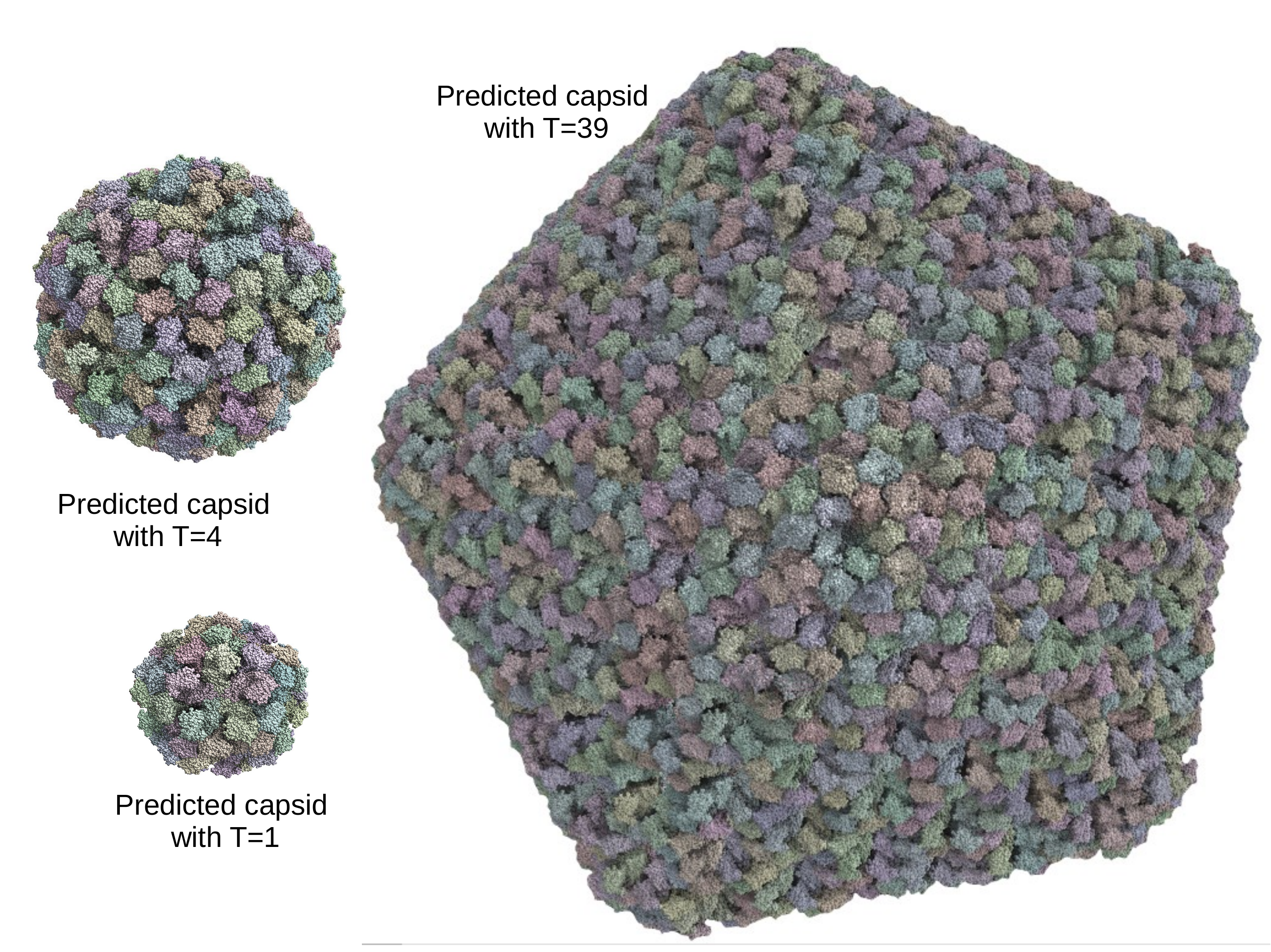}
\caption[Shells of different sizes using the same protein] {Shells of different sizes using the same protein}
\label{fig:virus:sizes}
\end{figure}

In Figure \ref{fig:assemblytheory:demo}(a), we present the results of modeling the Tobacco Necrosis Virus which has 60 proteins on its shell. We templated it based on a polyhedra with h=1, k=0 and the resulting computationally predicted shell had less that 5A RMSD error with respect to the known shell. This error is considered acceptable in molecular biology community. There is no topological errors. We had similar success in predicting the structure of Nudaurelia Capensis Virus which requires 240 proteins on its shell and the template polyhedra was constructed with h=2, k=0 (Figure \ref{fig:assemblytheory:demo}(b)). Finally, Figure \ref{fig:assemblytheory:demo}(c-d) show the outcome of predicting the shell of Rice Dwarf Virus which has 780 proteins. The layout for this can be either a polyhedra with h=3, k=1; or h=1, k=3; the latter is topologically incorrect if compared to the shell found in nature. Unfortunately, our geometric optimization algorithm and scoring model cannot discriminate between the two.

\subsubsection{Assembling Multiple Sized Shells Using the Same Building Block}

In Figure \ref{fig:virus:sizes} we show that our algorithm can easily produce shells of different sizes using the same building blocks. Here we used the same protein, but decorated tilings of different complexities and reported the highest scoring models for each size.

\section{Conclusion}
\label{sec:assemblytheory:conclusion}
We have characterized a new family of polyhedra with regular faces such that it is isotoxal, isohedral, and have exactly 2 types of vertices; as well as a dual family which is isogonal, isotoxal and have exactly 2 types of regular faces. We have shown that both of polyhedrons of these families generated by unfolding a regular polyhedron onto a lattice in a compatible way, thereby allowing the lattice vertices, edges and faces to etch out a tiling on the unfolded polyhedron, and finally folding it back again. Further, the compatible ways are specified using only a couple of integer parameters. We also provided a deterministic and efficient algorithm for generating such polyhedra of any size (determined by the two parameters). We have proved that our construction covers all possible polyhedron which satisfies the stated properties. When considering the geometric aspects of the generated polyhedra, we characterized the cases where the faces may become non-regular, and provided solutions for each case. 

Finally, we point out that our class of polyhedron is not similar to the known families like Catalan solids, Johnson solids, or Archimedean solids. Some Catalan solids like the tetrakis hexahedron, triakis octahedron, triakis icosahedron, rhombic dodecahedron, rhombic triacontahedron, or pentakis dodecahedron may seem like they satisfy the properties of \emph{almost-regular} polyhedron, but actually all of them violate either the global of the local symmetry conditions. Also Archimedean solids like the truncated cube can be generated by placing the triangles of a tetrahedron on a hexagonal lattice such that the corners of each triangle fall on the centers of 3 faces surrounding a single face. Many other Archimedian solids are isogonal and isotoxal, but are none of them (not even the truncated cube) are duals of of any \emph{almost-regular} polyhedron.

The characterization and construction would greatly aid computational analysis and modeling of dome or spherical shaped objects with symmetry and regular tiles, which is relevant for anti-viral drug design, designing nano-cages for drug delivery and cancer therapy, designing easy to assemble masonry structures etc. Additionally, the regular tilings we produce promises to be an interesting template for arbitrarily refined meshing of manifolds, by diffeomorphic mappings between the manifold and a sphere.

\end{document}